\documentclass[10pt]{article}
\usepackage[affil-it]{authblk}
\usepackage[margin=1in]{geometry}
\usepackage[bookmarks,hidelinks]{hyperref}
\usepackage{amsmath,listings,subfigure,amsthm,url,cite,epstopdf}
\newtheorem{theorem}{Theorem}[section]
\usepackage{color}
\usepackage{times}
\usepackage{inconsolata}
\usepackage{tikz}
\def\checkmark{\tikz\fill[scale=0.4](0,.35) -- (.25,0) -- (1,.7) -- (.25,.15) -- cycle;}
\definecolor{mygreen}{rgb}{0,0.2,0}
\definecolor{mygray}{rgb}{0.5,0.5,0.5}
\definecolor{mymauve}{rgb}{0.58,0,0.82}
\definecolor{mypurple}{rgb}{0.38,0,0.32}
\definecolor{myblue}{rgb}{0.1,0,0.32}
\newcommand{\costyle}{\footnotesize\ttfamily\bfseries}
\newcommand{\kwstyle}{\costyle\textcolor{myblue}}

\newcommand{\lt}{\left}
\newcommand{\rt}{\right}

\newcommand{\B}{\mathbf}
\newcommand{\CD}{\small\ttfamily\bfseries}

\begin{document}

\abovedisplayskip=4pt
\belowdisplayskip=4pt
\abovedisplayshortskip=0pt
\belowdisplayshortskip=2pt

\setlength{\pdfpageheight}{\paperheight}
\setlength{\pdfpagewidth}{\paperwidth}
\title{Sparse Tensor Algebra as a Parallel Programming Model}

\author{Edgar Solomonik and Torsten Hoefler \\
{\small solomonik@inf.ethz.ch \ \ \ \  \ \ \ \  \ \ \ \ \ \ htor@inf.ethz.ch \ \ \ }}
\affil{Department of Computer Science, ETH Zurich}

\maketitle

\begin{abstract}

Dense and sparse tensors allow the representation of most bulk data structures in computational science applications. We show that sparse tensor algebra can also be used to express many of the transformations on these datasets, especially those which are parallelizable. Tensor computations are a natural generalization of matrix and graph computations. We extend the usual basic operations of tensor summation and contraction to arbitrary functions, and further operations such as reductions and mapping. The expression of these transformations in a high-level sparse linear algebra domain specific language allows our framework to understand their properties at runtime to select the preferred communication-avoiding algorithm. To demonstrate the efficacy of our approach, we show how key graph algorithms as well as common numerical kernels can be succinctly expressed using our interface and provide performance results of a general library implementation.

\end{abstract}

\section{Introduction}
\label{sec:intro}

%
%
%
%
%
%
%

Vectors, which can be represented as arrays, are the most basic data structures in computer science.
Their formation and manipulation is most often done with loops, which can be parallelized so long as the iterations are independent.
Tensors, which can be directly represented as multidimensional arrays but also linearized into one-dimensional arrays, enable the programmer to express computations in their natural dimensionality.
While many programming libraries, including the basic C++ Standard Library provide primitives for array operations, the BLAS~\cite{lawson1979basic} and LAPACK~\cite{LAPACK} libraries and their concurrents provide a much richer set of operations on numerical vectors and matrices.
Many of these primitives have also been extended to sparse multidimensional arrays, most commonly, sparse numerical matrices and sparse graphs, enabling efficient abstractions for computations with irregular (non-dense) structure.
We seek to unite and generalize these array, graph, and matrix primitives by formulating them as algebras over (sparse) tensors.

By providing tensor abstractions equipped with the ability to transform as well as interact with other tensor objects, we seek to enable computational kernels to be wholly expressed in a high-level language.
The value of this paradigm lies not only in succinct and clear code, but also in the ability to abstract the data layout away from the programmer and enable automated parallelization over massively-parallel systems.
Our implementation of these abstractions takes the form of a C++ library rather than a compiler or programming language to enable easy integration and interoperability.
We start from an existing code, Cyclops Tensor Framework (CTF), which
has been shown to provide efficient distributed-memory dense tensor
abstractions for electronic structure
calculations~\cite{solomonik2014massively}, and extend its capabilities
in an attempt to service a much broader set of applications\footnote{The complete implementation of CTF and the benchmark codes presented in this paper are available at \url{https://github.com/solomonik/ctf}.}.
In particular, we enable the tensor objects in CTF to be sparse and their elements to be members of any type and algebraic structure.
Further, we generalize the notion of tensor summations and contractions using arbitrary elementwise functions, permitting the interaction of tensors with elements of two or three different types.

The combination of these two extensions is powerful, but still does not include some common primitive operations, such as graph traversals or matrix factorizations, although these could in principle be constructed using a sequence of CTF tensor operations.
Rather, the computations expressible using the natural language of tensor summations and contractions we provide have to be completely bulk synchronous in the sense that they are a combination of maps, independent operations, and reductions over sparse domains (our reductions are defined by the algebraic structure, unlike those of MapReduce~\cite{dean2008mapreduce}).
This limitation of expressibility is in fact also an advantage of our approach, as it becomes simple to write low-depth (efficiently parallelizable) programs and difficult to write polynomial-depth programs (which depth-first search and direct matrix factorizations are).
Further, our interface forgoes the need to write explicit loops, avoiding the associated correctness dangers of loop-carried dependencies and overhead of bounds checks.

In this limited and very succinct language, we are still able to express a broad range of complex algorithms in different application domains.
We limit ourselves to a few representative examples that leverage sparse tensors: iterative solvers for differential equations, shortest path computations, and electronic structure calculations (the classical consumer of higher-order tensor operations).
Despite using a Bulk-Synchronous Parallel (BSP) execution model~\cite{valiant1990bridging} and working with linearizable data-structures, the multidimensionality of tensors also permits us to express recursive algorithms via a correspondence of tensor dimensions to recursive levels.

We provide an implementation and performance results for the most critical types of sparse tensor operations, namely the summation of sparse tensors, and the contraction of sparse and dense tensors, but not contraction of a pair of sparse tensors or contraction into an output with a predefined sparsity (which we leave as future work).
Our implementation approach leverages and extends the existing CTF infrastructure for mapping, decomposing, redistributing, and contracting dense tensors, but now with sparse layouts and with account of the number of nonzeros.
That is, we do not attempt to infer the nonzero structure of the tensors and perform graph or hypergraph partitioning to achieve an optimal decomposition, but rather leverage the cyclic CTF layout to randomize the nonzero structure.
The result yields algorithms that are optimal for random nonzero structure and, with high probability, for any other nonzero structures.

Our analysis and performance evaluation show that the resulting algorithms reduce both the computation and, in many cases, the communication costs in proportion to the fraction of nonzeros in the tensors.
We show that for a simple sparse matrix benchmark and a more involved use of sparse tensors in an electronic structure method, sparsity improves the time to solution and helps weak scalability significantly.
We further analyze the performance of the computation of all-pairs
shortest-paths of a dense graph via path-doubling, where the use of sparsity enables asymptotically better computation cost~\cite{tiskin_apsp} and better scalability in practice.
While we leverage MKL sparse matrix routines for the first two benchmarks, the third uses our own reference kernels, as it requires the application of a user-defined elementwise function on integer types.

\section{Previous Work}
\label{sec:prev}


Dense matrix abstractions are one of the most well-studied and used primitives in numerical computations.
Sequential frameworks such as the BLAS~\cite{lawson1979basic}, as well as parallel libraries such as ScaLAPACK~\cite{SCALAPACK}, and Elemental~\cite{elemental1} have demonstrated the ability of matrix abstractions to service costly components of a wide range of applications with high-performance kernels.
The hierarchy of the BLAS also demonstrates that higher level abstractions provide more and more performance, e.g., BLAS 3 with respect to BLAS 2.

A number of research efforts have attempted to further raise the level of abstraction and provide primitives for tensors operations.
Frameworks such as Libtensor~\cite{JCC:JCC23377} have shown that performance close to machine peak may be achieved generally for tensor contractions on shared memory machines.
Tensor Contraction Engine (TCE)~\cite{doi:10.1021/jp034596z,1386652} and Global Arrays~\cite{springerlink_ga1} provided an approach to automatic factorization of tensor expressions and generation of code to execute tensor contractions in distributed memory.
Alternatively, a number of frameworks have provided a library-based approach to distributed memory tensor contractions ~\cite{rajbhandari2013framework,Lai:2013:FLB:2503210.2503290,solomonik2014massively}.
These frameworks have all primarily targeted the domain of electronic structure calculations, in some cases providing support for symmetry and block-sparsity typical to these computations.
Further research has exploited the low-rank structure presented in electronic structure calculations, and a recent effort has provided a low-rank block-tensor abstraction for distributed memory machines~\cite{2015arXiv150900309C}.
General sparse tensor frameworks are also not new and have been implemented and applied to sequential electronic structure calculations~\cite{Kats_sp_tensor2013,doi:10.1080/00268971003662896}.

Outside of the electronic structure domain, studies have primarily focused on sparse matrices and graphs.
The Combinatorial BLAS (CBLAS)~\cite{bulucc2011combinatorial}, Parallel BGL~\cite{BGL2005}, and Pregel~\cite{malewicz2010pregel} have all provided key primitives for sparse graph computations.
A number of these graph computations are isomorphic to sparse matrix computations, differing from sparse matrix multiplication only in the elementwise operations.
The kernel of sparse-matrix vector multiplication is so ubiquitous in numerical computations and well-studied from a performance perspective, that we forgo any survey of the literature on this topic (see~\cite{Buluc:2009:PSM:1583991.1584053} and references therein).
However, the problem of multiplication of a sparse matrix by a dense matrix, a critical kernel in numerical computations, has been studied noticeably less.
One classical approach to this problem is the outer-product algorithm~\cite{Kruskal1989135}.
Parallel algorithms for multiplying two sparse matrices have also been recently studied in terms of their communication cost~\cite{doi:10.1137/110848244,Ballard:2013:COP:2486159.2486196}.
In this paper, we move towards studying algorithms in a more general setting of sparse tensors, focusing on the problem most relevant to applications, the contraction of sparse tensors with dense tensors.
Computationally, the key part of this operation can always be reduced to the multiplication of a sparse matrix and a dense matrix, a problem we analyze in detail in Section~\ref{sec:algs}.

\section{Tensor Algebra Interface}
\label{sec:stalg}
%

Sparse tensors have a correspondence to hypergraph dataset representations, where each hyperedge contains $d$ vertices if the tensor is of $d$th order.
By providing a general interface for basic algebraic operations on tensors, we aim to allow succinct expression of bulk operations on multidimensional arrays and hypergraphs (graphs for $d=2$).
We allow tensor elements to be defined on arbitrary algebraic structures, which express the basic properties (associativity, distributivity, zeros) of the desired elementwise tensor operations.
Each tensor is not restricted to elementwise operations defined by its algebraic structure, auxiliary user-defined functions may also be defined and applied. 
Such an approach is not entirely new, the Julia programming language~\cite{2012arXiv1209.5145B} leverages general types and elementwise functions for matrix computations.
The algebraic properties of the elementwise functions define the space of algorithms and optimizations that are permitted when executing these operations over entire tensors.
Finally, we employ an index interface that allows the definition of loop nests acting on one or two tensor operands and producing a single tensor output.
The tensor indices may be thought of as loop indices, with the total number of unique indices in an expression defining a loop nest of that order, with the elementwise operations applied in the innermost loop.
Indices that are omitted in the output are summed over, in line with the Einstein summation notation used in chemistry and physics.

\subsection{Algebraic Structures}
\label{subsec:algstrct}

Algebraic structures allow us to define the types of tensor elements and specify their properties.
Tensor summation and contraction are derived on top of the operations specified in these algebraic structures.
Our interface provides five types of structures, whose properties are summarized in Table~\ref{tab:algstr}.
\begin{table}[htbp]
\centering
\begin{tabular}{c | c | c | c | c | c}
algebraic structure & add. op. & add. id. & add. inv. & mul. op. & mul. id. \\
\hline
{{set}} & & & & & \\
\hline
{{monoid}} & \checkmark & \checkmark & & & \\
\hline
{{group}} & \checkmark  & \checkmark & \checkmark & & \\
\hline
{{semiring}} & \checkmark & \checkmark & & \checkmark & \checkmark \\
\hline
{{ring}} & \checkmark & \checkmark & \checkmark & \checkmark & \checkmark 
\end{tabular} 
\caption{Summary of algebraic structure properties: operators (add. op. and mul. op.), inverses (add. inv.), and identities (add. id., mul. id.).}
\label{tab:algstr}
\end{table}
As it is also possible to define auxiliary functions, algebraic structures that do not have identity elements, namely semigroups (monoid without an additive identity) and rngs (ring without a multiplicative identity), do not need to be explicitly supported.
The action of the semigroup operator may be fully expressed by a function applied to pairs of elements of a set and the action of the multiplicative operator in an rng may be expressed as a function on pairs of elements of a group.
We assume associativity of the additive operators and distributivity of the multiplicative operators defined for these algebraic structures, but allow the user to specify these properties for auxiliary functions. 
The user can also define algebraic structures that satisfy these properties within only a subdomain of elements or only approximately (e.g., floating point arithmetic).
Currently, no static or runtime checks of these properties are done, so it is up to the user to ensure the computation is specified correctly.

The most basic structure, a set, simply holds the C++ type of the tensor elements.
Our implementation is restricted to statically-sized types, which makes iteration and communication of elements straight-forward.
More advanced algebraic structures additionally take elementwise functions as parameters for operators and specific elements as zeros.
The elementwise functions may be provided as function pointers or C++11 Lambdas, and the natural default operators/zeros are provided for basic types.
For example, the default tensors are created on a ring algebraic structure equivalent to the following {\kwstyle Ring} {\CD r},
\begin{lstlisting}
     Ring<> r(0.0,                                    // additive identity
              [](double a, double b){ return a+b; },  // additive operator
              MPI_SUM,                                // MPI Op for additive operator
              1.0,                                    // multiplicative identity
              [](double a, double b){ return a*b; }); // multiplicative operator
\end{lstlisting}
As another example, we define an algebraic structure used for shortest path computations on graphs with integer weights, the tropical semiring, as follows,
\begin{lstlisting}
    Semiring<int> ts(INT_MAX/2,                             // additive identity
                     [](int a, int b){ return min(a,b); },  // additive operator
                     MPI_MIN,                               // MPI Op for additive operator
                     0,                                     // multiplicative identity
                     [](int a, int b){ return a+b; });      // multiplicative operator
\end{lstlisting}
In the tropical semiring, the addition operator is an integer minimum, while the multiplication operator is integer addition.
The identity element of the multiplication operator is provided as $0$, while the identity element of the addition operator is provided as half of the largest representable integer.
Using half of the integer max as the additive identity prevents integer overflow when adding (multiplying on the tropical semiring) two additive identities, avoiding the need for an overflow check in the actual operator for the tropical semiring computations presented in this paper. 

\subsection{Tensors}

Tensor objects may be defined as ordered collections of elements of some algebraic structure.
We denote a tensor as $\mathbf{T}$ and its elements as $T_{ijk..}$.
Each element of a tensor is given by some configuration of a set number of indices.
This number of indices is referred to as the order of the tensor, while the size of the range of the indices are its dimensions.
For instance, a scalar has order zero, a vector has order one, and a matrix has order two.

Each tensor is distributed across a set of processors defined by a {\kwstyle{World}}.
In our implementation each {\kwstyle{World}} corresponds to a set of MPI processes.
For instance, a {\kwstyle{World}} containing all available processes may be defined as
\begin{lstlisting}
    World dw(MPI_COMM_WORLD);
\end{lstlisting}



The simplest and most commonly used tensors are those of low order.
While we provide special interfaces for {\kwstyle{Scalar}}, {\kwstyle{Vector}}, and {\kwstyle{Matrix}}, they are all simply instances of {\kwstyle{Tensor}}.
The constructors for all these objects take parameters in the following order
\begin{enumerate}
\item element type (template parameter)
\item the order and dimensions of the tensor
\item sparsity and symmetry attributes of the tensor
\item the set of processes over which the tensors are distributed
\item the algebraic structure of the tensor elements
\end{enumerate}
The parameters may be omitted when they take on the default, e.g., the element type is `double', no sparsity or symmetry attributes, the tensor is distributed over all processes, or the algebraic structure is the standard addition/multiplication over a basic type.
The first integer parameter for {\kwstyle{Vector}} defines the dimension of the vector.

The symmetry of a matrix may be specified as one of
NS (nonsymmetric), SH (symmetric without diagonal), SY (symmetric with diagonal), and AS (antisymmetric).
These can be combined (binary or) with an attribute for sparsity, SP.
For instance, we can define an undirected graph with $n$ nodes and integer edge weights as a sparse $n\times n$ symmetric matrix with no diagonal on the tropical semiring,
\begin{lstlisting}
    Matrix<int> F(n,n,SP|SH,ts);
\end{lstlisting}


Tensors of arbitrary type and order are defined via the {\kwstyle{Tensor}} class.
In our implementation, this templated (by element type) class derives from an internal non-templated class, which implements effectively all functionality, ensuring that instantiation of tensors of new type takes little compilation time.
For arbitrary-order tensors, the symmetry attribute specification is more complex, and allows for the definition of partially symmetric tensors.
An array of size equal to the order of the tensor of symmetry attributes must be provided, where the $i$th element specifies the symmetry relation of the $i$th tensor dimension with the $(i+1)$th.
For example, we can construct a sparse fourth order tensor often used in quantum chemistry, which has two pairs of antisymmetric indices as
\begin{lstlisting}
    int dims[] = {nv,nv,no,no};
    int syms[] = {AS,NS,AS,NS};
    Tensor<> V(4,dims,SP,syms);
\end{lstlisting}
This tensor would contain double-precision floating-point elements with the standard ring operators, and have the dimensions {\CD nv}-by-{\CD{nv}}-by-{\CD{no}}-by-{\CD{no}} with the first index being antisymmetric with the second (specified by the position of the first AS), and the third being antisymmetric with the fourth (specified by the position of the second AS).

While creation of dense tensors immediately maps and allocated the data (and sets to the additive identity, when the algebraic structure has one), a newly created sparse tensor contains no data.
Data is read and written to sparse and dense tensors bulk synchronously via index-value pairs.
Special functions provide access to data local to the processor and to certain specified distributions.

\subsection{Tensor Operations}

Having defined the algebraic structure and properties of the representation of the tensors, we can now define various algebraic operations.
We focus on defining operations on one or two tensors, as functions of multiple tensors can be deconstructed into pairwise tensor functions (albeit with a potential need for larger intermediate tensors).
We first recall how Einstein summation notation may be used to express tensor summations and contractions over the algebraic structure.
We then study two types of specially-defined operators, `functions' that take one or two elements as operands and output a new element, and `transforms' that modify an element, based zero, one, or two other elements as operands.
These functions are more general than those specifiable in algebraic structures, as all operands/outputs are permitted to be of different type.

\subsubsection{Indexed Tensors}

To relate tensors to one another, we can assign any order $k$ tensor, $k$ characters as `indices' for each of its $k$ ways.
Such indices are commonly used in tensor mathematics, especially in chemistry and physics, and may also be thought of as for loop variables in the context of operations.
When the same index appears in two tensors being operated on, or in a tensor and a result, these indices are matched and specify the semantics for the operation.
When an index appears two or more times in the same tensor, it specifies that the operation should touch only the corresponding diagonal of the tensor.
The use of pure indices for expression of tensor summations and contractions is referred to as Einstein summation notation.
This indexed interface has been previously presented for CTF~\cite{solomonik2014massively}, and similar interfaces have been used in other tensor contraction libraries~\cite{JCC:JCC23377}.
We start by recalling the basics of the index notation, then show how the notation can be combined with more arbitrary elementwise functions.

Our interface creates indexed tensor objects from tensor objects by overloading the bracket operator for tensors, e.g., we can assign the {\CD V} tensor created above indices as follows,
\begin{lstlisting}
    Idx_Tensor V["abij"];
\end{lstlisting}
It is often convenient to inline this notation directly into the tensor expression.
For example, given tensor {\CD W} with the same dimensions and symmetry as {\CD V}, we can add it to {\CD W} via the operation,
\begin{lstlisting}
    V["abij"]+=W["abij"];
\end{lstlisting}
When indices appear only in the operands, an implicit summation over this index is implied, while when an index appears exclusively in the output tensor, it is implied that the result is replicated (mapped) over this index.
For example the following reduction and map operations:
\[q = \sum_{i=1}^n v_i, \quad 
 \forall i\in\{1,\ldots, n\}, b_{i} = \sum_{j=1}^n A_{ij}, \quad
 \forall i \in \{1,\ldots, n\}, z_i = 42, \quad 
 \forall i,j\in\{1,\ldots, n\}, F_{ij} = \sum_{k=1}^n G_{kj},\]
are expressed very similarly in our interface once the tensors are defined appropriately,
\begin{lstlisting}
    q[""] = v["i"];       b["i"] = A["ij"];       z["i"] = 42;       F["ij"] = G["kj"];
\end{lstlisting}
Additionally, the interface supports iteration over diagonals.
Below, we show three examples: defining an identity matrix, scaling the diagonal of a matrix, and computing the sum of the superdiagonal of a third-order tensor,
\begin{lstlisting}
    I["ii"] = 1.0;                 A["jj"] *= 3.0;                 double s = T["iii"];
\end{lstlisting}
The scalar $s$ is automatically cast to the type of an order-zero tensor.

For basic types, it is possible to cast scalars to tensors with no indices, e.g., we can compute the square of the Frobenius norm (2-norm) of {\CD V} via the command,
\begin{lstlisting}
    double nrm_sq = V["abij"]*V["abij"];
\end{lstlisting}
although we also provide a special {\kwstyle norm2()} function for this faculty.
Here all the four indices $a,b,i,j$ are summed over, which is implicitly inferred as these indices do not appear in the output.
If {\CD V} has dimensions $n\times n\times n \times n$, the above line of code computes
\[q = \sum_{a=1}^n\sum_{b=1}^n\sum_{i=1}^n\sum_{j=1}^n V_{abij}^2.\]
This interface supports typical tensor contractions such as (with index and summation limits omitted for brevity),
\[\forall i,\quad z_i = \sum_j W_{ik}\cdot v_k, \quad\quad 
  \forall i,j, \quad C_{ij} = \frac{1}{2}\sum_k A_{ik}\cdot B_{kj}, \quad\quad 
  \forall a,b,i,j, \quad F_{abij} = F_{abij} + \sum_k G_{ikab}\cdot H_{kj},\]
via simple one line commands that closely correspond to the mathematical formulation:
\begin{lstlisting}
    z["i"] = W["ik"]*v["k"];   C["ij"] = .5*A["ik"]*B["kj"];   F["abij"] += G["ikab"]*H["kj"];
\end{lstlisting}
In the above tensor contractions each index appears in exactly two tensors, consistent with the typical definition of contractions.
However, as in summations, within our interface it is also possible to perform tensor operations where indices appear in exclusively one tensor, corresponding, as before, to reductions or maps.
Further, it is possible for indices to appear in all three tensors, which specifies a set of independent problems over this index.
For example, the Hadamard matrix product, defined by
\(\forall i,j, T_{ij} = V_{ij}\cdot W_{ij}\)
is easily expressible via our interface,
\begin{lstlisting}
    T["ij"] = V["ij"]*W["ij"];
\end{lstlisting}
More generally, indexed tensor summations and contractions can be interpreted by imagining a set of nested for loops over every unique index with the tensor expression appearing in the innermost loop as an (accumulation) operation on elements of multidimensional arrays.

\subsubsection{Tensor Functions and Transforms}

We allow the algebraic structure operator to be replaced (or the algebraic structure to be extended), by allowing basic user-defined functions.
This syntax allows the user to succinctly express tensor transformations that modify each element independently, as well as combine pairs of element of two tensors to produce a third.
There are two signatures for user defined functions:
\begin{lstlisting}
    (type_B) <- (type_A),
    (type_C) <- (type_A, type_B).
\end{lstlisting}
When the types are the same, the latter function signature is the same as that of the addition and multiplication operators defined for algebraic structures.
When the tensor expression requires summation of a set of function outputs, the addition operator of the output tensor algebraic structure is used.

As an example, a vector of forces {\CD F} may be formed as a set of partial sums of interactions of particles stored in vector {\CD P}.
To do this, first a {\kwstyle Monoid} algebraic structure is created, which requires defining an MPI reduction operator for the summation of forces.
Then, once the vector {\CD P} is populated with data (local data may be written to tensor bulk-synchronously via the {\kwstyle write} function), the forces are calculated and accumulated in one command.
\begin{lstlisting}
    struct force{
      double x, y; 
      force(double x0, double y0){ x=x0; y=y0; }
    };
    MPI_Op op_add;
    MPI_Op_create([](void * a, void * b, int * n, MPI_Datatype *){
                    for (int i=0; i<*n; i++){
                      ((force*)b)[i].x += ((force*)a)[i].x;
                      ((force*)b)[i].y += ((force*)a)[i].y;
                    }
                  }, 1, &op_add);
    Monoid<force> mf(force(0.,0.), [](force a, force b){ return force(a.x+b.x,a.y+b.y); }, op_add);
    Vector<force> F(n, mf);
    
    struct particle{ double x, y, px, py, A; };
    Vector<particle> P(n, Set<particle>());
    P.write( ... ); // input local particle data and indices **/ 

    force interact(particle a, particle b){
      double dx = a.x-b.x;    double dy = a.x-y.x;    double f21 = a.A*b.A/dx*dx*dy*dy;
      return force(f21*dx, f21*dy);
    }
    Function<particle,particle,force> f(&interact);
    F["i"] += f(P["i"],P["j"]); // in parallel compute for i = 1 to n, F_i = sum_j f(P_i, P_j)
\end{lstlisting}

Transforms provide a more powerful abstraction, that allows accumulation and modification to existing elements, overriding the addition operator of the output tensor algebraic structure.
There are three signatures for transforms:
\begin{lstlisting}
    void (&type_A), 
    void ( type_A, &type_B),
    void ( type_A,  type_B, &type_C).
\end{lstlisting}
In each case, the last value is passed by reference and should be modified within the transform function.
As an example, a set of forces may be integrated to update a particle's location, via the following transform.
\begin{lstlisting}
    Transform<force,particle> t([] (force f, particle & p){ p.px += p.A*f.x; p.py += p.A*f.y; });
    t(F["i"],P["i"]);
\end{lstlisting}
This interface syntax allows arbitrary-depth nested loops to be executed in parallel via one line of code, e.g., the command,
\begin{lstlisting}
    ((Transform<>)([] (double a, double b, double & c){ c=c/(a+b); }))(A["i"],B["j"],C["ij"]);
\end{lstlisting}
amounts to executing the following code over distributed arrays,
\begin{lstlisting}
    for (int i=0; i<n; i++){
      for (int j=0; j<n; j++){
        C[i,j] = C[i,j]/(A[i]+B[j]);
      }
    }
\end{lstlisting}
By default the framework assumes transforms and functions are distributive.
The non-distributive case, which also has important use-cases, is not yet supported by our implementation.

\subsection{Implementation of the Interface}

The realization of arbitrary elementwise tensor functions and transforms does not present an algorithmic challenge, but poses an interesting practical implementation challenge.
In particular, functions that combine/interact/produce tensors of different types imply that the entire framework cannot simply be templated and instantiated for a given type (the approach taken by CTF previously).
Instead, the overwhelming bulk of the logic of the framework is implemented using runtime-specification of the type, i.e., the number of bytes needed for each tensor element suffices for the execution of the mapping and redistribution logic.
A lightweight templated layer is built on top of this, providing a user-friendly typed interface, while keeping compilation time small even when triply-templated user-defined functions are being applied.
This approach is a significant departure from and extension of the standard four numerical types provided by BLAS and LAPACK libraries.

\section{Applications}
\label{sec:apps}

%
%

We consider a few common algorithms from three computational domains that can be formulated succinctly as tensor operations.
We start with algorithms for some general sparse iterative methods, then consider shortest path computations in graphs, and finally an electronic structure method from quantum chemistry.

\subsection{Sparse Iterative Methods}

Sparse iterative methods are widely used to obtain numerical solutions to differential equations from a variety of problem domains.
Sparse-matrix vector multiplication is a ubiquitous primitive for such methods.
This primitive is well-studied and available in numerous existing sequential and parallel numerical libraries (see~\cite{Buluc:2009:PSM:1583991.1584053} and citations therein).
To illustrate a simple example usage of our interface, we present a Jacobi iteration for an arbitrary $n\times n$ sparse matrix of double precision floating pointer numbers, $\B A$, and a right-hand side $\B b$.
Jacobi iteration solves the matrix equation \(\B A\B x=\B b\) using the iterative scheme,
\[\forall i\in [1,n], \quad x^{l+1}_i=(1/A_{ii})\cdot(b_i-\sum_{i=0,i\neq j}^{n} A_{ij}\cdot x^l_j),\]
with guaranteed convergence when $A$ is diagonally dominant.
The below code (where {\CD n} is the dimension of {\CD b} and {\CD dw} is the {\kwstyle{World}} we previously defined) implements Jacobi iteration using our interface.
\begin{lstlisting}
  void Jacobi(Matrix<> & A, Vector<> & b, World & dw, int n){
    Vector<> x(n,dw);
    Vector<> d(n,dw);
    Vector<> r(n,dw);
    Matrix<> R(n,n,SP,dw);
    d["i"] = A["ii"];  // set d to be diagonal of A
    // invert each element of d
    (Transform<> ([] (double & d){ d=1./d; }))(d["i"]);
    R["ij"] = A["ij"]; // set R = A
    R["ii"] = 0;       // set the diagonal of R to zero (R=A-diag(d))
    do { 
      // compute x = diag(d)*(Ax+b)
      x["i"]  = -R["ij"]*x["j"];
      x["i"] +=  b["i"];
      x["i"] *=  d["i"];
  
      // compute residual r = b - Ax
      r["i"]  =  b["i"];
      r["i"] -=  A["ij"]*b["j"];
    } while (r.norm2() > 1.E-6); // repeat until convergence
  }
\end{lstlisting}
The application of sparse tensor contractions to iterative methods is not limited to simple sparse-matrix vector schemes such as Jacobi iteration.
For instance, multigrid interpolation and restriction operators may be formulated as multiplication of sparse matrices~\cite{doi:10.1137/110848244}.
%

\subsection{Graph Algorithms}
\label{sec:apps:ga}

The duality of graphs and sparse matrices is well-known and much previous research has been done on algebraic formulations of graph algorithms.
A number of algebraic formulations of graph algorithms were surveyed in~\cite{doi:10.1137/1.9780898719918}, including fundamentals such as the Bellman--Ford algorithm for single-source shortest-paths~\cite{bellman1958routing,ford1958network}, the Floyd--Warshall algorithm for all-pairs shortest-paths~\cite{Floyd:1962,Warshall:1962}, and Prim's algorithm for construction of minimal spanning trees~\cite{6773228}.
While we focus on shortest-paths algorithms, we note that these can be extended to the use of more interesting algebraic structures such as the geodetic semiring, which allows the computation of betweenness centrality~\cite{batagelj1994semirings,brandes2001faster}, an important measure in community structure analysis.

The shortest distances from one node to all others may be computed using Dijkstra's algorithm when the edge-lengths are positive.
However, Dijkstra's algorithm achieves this by relaxing a different subset of edges every iteration, with dependence on the previous iteration, making parallelizations inefficient.
On the other hand, the Bellman--Ford~\cite{bellman1958routing,ford1958network} algorithm, which relaxes all edges at every iteration, may be expressed simply as a matrix-vector multiplication on the tropical semiring.
For any graph $G=_(V,E)$, the adjacency matrix $\B A$ encodes the edge weights, i.e., $A_{ij}$ is the weight of edge $(j,i)\in E$ or $A_{ij}=\infty$ if $(j,i)\notin E$, and $A_{ii}=0$.
If the source vertex is the zeroth vertex, we can initialize a vector of distances $\B{P^{(0)}}=(0,\infty,\ldots)$, of dimension $|V|$.
Bellman-Ford computes the iterative scheme 
\[P^{(r)}_i=\min_j (A_{ij} + P_j^{(r-1)}),\]
starting from $r=1$ and until convergence ($\B{P^{(r)}}=\B{P^{(r-1)}}$), which will be reached by $r=|V|$, unless the graph contains negative-weight cycles.
The following examples uses CTF to compute Bellman--Ford via matrix-vector multiplication on the tropical semiring, {\CD ts} (defined in Section~\ref{subsec:algstrct}).
\begin{lstlisting}
  // Input: potentially sparse adjacency matrix A and initial distances P
  // Output: shortest distances P=P+A*P+A*A*P+...
  // return false if there are negative cycles, true otherwise
  template <typename t>
  bool Bellman_Ford(Matrix<t> A, Vector<t> P, int n){
    Vector<t> Q(n);
    int r = 0;
    do { 
      if (r == n) return false;      // exit if we did not converge in n iterations
      else r++;
      Q["i"]  = P["i"];              // save old distances
      P["i"] += A["ij"]*P["j"];      // update distances 
    } while (P.norm1() < Q.norm1()); // continue so long as some distance got shorter
    return true;
  }
\end{lstlisting}
The function in the example above computes single-source shortest-paths, provided that {\CD P} and {\CD A} are defined on the tropical semiring {\CD ts} and that {\CD P} is initialized to zero for the source node index, and infinity (additive identity for the tropical semiring ) everywhere else.
We make the function templated, allowing it to be used for integer, floating-point, or custom-type edge weights.
Each iteration of the while loop, relaxes all edges, and stores the updated shortest paths (with up to {\CD r} hops) in {\CD P}.
The vector {\CD Q} is defined in order to perform a convergence check, which allows early termination or reports a negative cycle in the graph if convergence is unachievable.

We further consider the computation of all-pairs shortest-paths (APSP); we compute only distances, but the extension to paths is trivial.
The Floyd--Warshall algorithm~\cite{Floyd:1962,Warshall:1962} is the standard work-efficient algorithm for computing the distance matrix for an arbitrary graph.
However, while work efficient, it is polynomial depth (it is the tropical semiring equivalent of Gaussian elimination).
An alternative logarithmic-depth algorithm for computing APSP is based on path-doubling, which computes the closure of the adjacency matrix over the tropical semiring ({\CD ts}) by repeatedly squaring the matrix:
\begin{lstlisting}
  // Given n-by-n matrix A compute (I+A)^n
  template <typename t>
  void path_doubling(Matrix<t> A, int n){
    for (int l=1; l<n; l=l<<1){
      A["ij"] += A["ik"]*A["kj"];
    }
  } 
\end{lstlisting}
However, pure path doubling is not work efficient, it requires $O(n^3\log(p))$ operations instead of the $O(n^3)$ required by Floyd--Warshall.
A workaround for this was presented by Tiskin~\cite{tiskin_apsp}.
The main idea of Tiskin's path doubling algorithm is that all shortest paths of $2l$ hops (going through $2l$ edges) have either fewer than $l$ hops or contain a shortest path of exactly $l$ hops plus another path of no more than $l$ hops.
Thus it suffices to pick out all paths of exactly $l$ hops and multiply by this sparse matrix.
Tiskin's algorithm adaptively selects the number of hops that the sparse matrix can contain between $l/2$ and $l$, to ensure that this matrix has few nonzeros.
We present a slightly simplified version of the algorithm by picking shortest paths of length $l$ at each step, losing the guarantee of $O(n^3)$ cost, but only for very exotic graphs.
\newpage
\begin{lstlisting}
  // Given n-by-n matrix A compute (I+A)^n over an idempotent semiring
  void Tiskin_path_doubling(Matrix<int> A){
    // struct for path with w=path weight, h=#hops
    struct path {
      int w,h;
      path(int w_, int h_){ w=w_; h=h_; }
      path(){};
    };
    MPI_Op opath;
    ... // define reduction opath as elementwise min
    
    // ts semiring with hops carried by winner of min
    Semiring<path> p(path(INT_MAX/2,0), 
                     [](path a, path b){ 
                       if (a.w<b.w) return a; 
                       else return b; 
                     },
                     opath,
                     path(0,0),
                     [](path a, path b){ 
                       return path(a.w+b.w, a.h+b.h); 
                     });
    Matrix<path> P(n,n,dw,p);     // path matrix to contain distance matrix
    P["ij"] = ((Function<int,path>)([](int w){ return path(w,1); }))(A["ij"]);
    
    Matrix<path> Pl(n,n,SP,dw,p); // sparse path matrix to contain all paths of l hops
    
    for (int l=1; l<n; l=l<<1){
      // let Pl be all paths in P consisting of l hops
      Pl["ij"] = P["ij"];
      Pl.sparsify([=](path p){ return (p.h == l); });
      //            each   shortest path of up to 2l hops either 
      // (1)          is a shortest path of up to  l hops or
      // (2) consists of a shortest path of up to  l hops
      //             and a shortest path of        l hops
      P["ij"] += Pl["ik"]*P["kj"];
    } // P is the distance matrix
    A["ij"] = ((Function<path,int>)([](path p){ return p.w; }))(P["ij"]);
  }
\end{lstlisting}
This path doubling approach reduces the computation and communication bandwidth costs with respect to the dense approach as the number of nonzeros in the sparse matrix {\CD Pl} is expected to decrease geometrically with each shift of {\CD l}. 

\subsection{Electronic Structure Calculations}
\label{sec:apps:esc}

To provide an example with the use of higher-order tensor contractions, we return to the motivating application domain for Cyclops Tensor Framework, electronic structure calculations.
The modeling of electronic correlation accomplished within post Hartree-Fock methods is most often done via direct or iterative numerical schemes consisting of tensor contractions.
Almost universally, these contractions involve the two-electron integral tensor $V$ (often separated by spin-cases to reduce storage and cost).
As two-electron interactions decay cubically with distance, a localized set of orbitals (basis functions) enables accurate numerical solution of electronic correlation models without taking into account all interactions~\cite{HF_screening_1995,LCCSD_2001,doi:10.1080/00268971003662896,ChowHF2015}.

When small interactions are screened (ignored, but often accounted for by error correcting smoothing terms), the resulting two-electron integrals take on a sparse representation.
We study one of the simplest post-Hartree-Fock methods, third-order M{\o}ller-Plesset perturbation theory (MP3)~\cite{moller1934note}.
We employ the tensor contractions used within the Aquarius framework~\cite{solomonik2014massively} to compute the correction to the energy (for a derivation of similar MP3 equations, see~\cite{BarlettMP3_1975}),
\begin{align*}
E_\mathrm{MP3} = \sum_{abij} \frac{(V^{ab}_{ij})^2}{-\epsilon_a-\epsilon_b+\epsilon_i+\epsilon_j}\bigg(1&+V^{ij}_{ab}+\sum_f F^{a}_fT^{fb}_{ij} - \sum_n F^{n}_iT^{ab}_{nj} \\
&+\frac{1}{2} \sum_{e,f}V^{ab}_{ef}T^{ef}_{ij}+\frac 12 \sum_{m,n}V^{mn}_{ij}T^{ab}_{mn}-\sum_{e,m}V^{am}_{ei}T^{eb}_{mj}\bigg),
\end{align*}
given first-order energies $\B \epsilon$, one-electron and two-electron interaction tensors $\B F$ and $\B V$, as well as forming, $\B T$ as
\[ T^{ab}_{ij} = \frac{V^{ab}_{ij}}{-\epsilon_a-\epsilon_b+\epsilon_i+\epsilon_j}.\]
This set of tensor contractions can be reduced to operate on smaller tensors by `spin integration', a technique that accounts for the spin-orientation of electrons and orbitals whose interactions are represented by $\B F$ and $\B V$.
The spin-integrated method can be expressed in terms of the vectors 
{\CD Ea}, {\CD Ei} (blocks of $\B \epsilon$), the
matrices {\CD Fab}, {\CD Fij} (blocks of $\B F$), 
and the sparse fourth order tensors 
{\CD Vabij}, {\CD Vijab}, {\CD Vabcd}, {\CD Vijkl}, {\CD Vaibj} (blocks of $\B V$).
The dimensions of these tensors correspond to occupied orbitals (electrons) or virtual orbitals, of which there are $m$ and $n$, respectively.
When dense, the integral tensors themselves may be computed with cost $O(n^4)$, while the MP3 computation has cost $O(m^2n^4)$.
When all the integral tensors tensors are sparse, these costs drop proportionally to the number of nonzeros.
\begin{lstlisting}
  double MP3(Tensor<> Ei, Tensor<> Ea, Tensor<> Fab, Tensor<> Fij,
             Tensor<> Vabij, Tensor<> Vijab, Tensor<> Vabcd, Tensor<> Vijkl, Tensor<> Vaibj){
    // compute the denominator tensor, D_abij = 1./(-e_a-e_b+e_i+e_j)
    Tensor<> D(4,Vabij.lens,*Vabij.wrld);
    D["abij"] += Ei["i"]; 
    D["abij"] += Ei["j"]; 
    D["abij"] -= Ea["a"]; 
    D["abij"] -= Ea["b"]; 
    ((Transform<>)([](double & b){ b=1./b; }))(D["abij"]);

    // form T_abij = V_abij/(-e_a-e_b+e_i+e_j)
    Tensor<> T(4,Vabij.lens,*Vabij.wrld);
    T["abij"] = Vabij["abij"]*D["abij"];

    // compute Z_abij = V_ijab + <F,T> + <V,T>
    Tensor<> Z(4,Vabij.lens,*Vabij.wrld);
    Z["abij"]  = Vijab["ijab"];
    Z["abij"] += Fab["af"]*T["fbij"];
    Z["abij"] -= Fij["ni"]*T["abnj"];
    Z["abij"] += 0.5*Vabcd["abef"]*T["efij"];
    Z["abij"] += 0.5*Vijkl["mnij"]*T["abmn"];
    Z["abij"] -= Vaibj["amei"]*T["ebmj"];

    // update T_abij = T_abij + D_abij*Z_abij 
    T["abij"] += Z["abij"]*D["abij"];

    // return sum(abij) T_abij*V_abij
    double MP3_energy = T["abij"]*Vabij["abij"];
    return MP3_energy;
  }
\end{lstlisting}
The most expensive (sixth order) contractions in the method are the last three terms contributing to {\CD Z}.
However, in practice, the applications of the denominator, {\CD D} as well as the other lower order contractions must also be executed efficiently.
We will evaluate the sparse tensor contraction algorithms presented in the following section, by benchmarking the MP3 and shortest-path examples.

\section{Algorithms and Analysis}
\label{sec:algs}


Our algorithms for sparse tensor summation and contraction build directly on those implemented in CTF for dense tensors~\cite{solomonik2014massively}.
Our approach is sparsity-oblivious, in the sense that it uses the same decomposition for all sparsity distributions, taking into account only the nonzero count.
In this way, the mapping logic is general and inexpensive, but at the cost of not finding optimal decompositions for certain sparsity structures.
The parallel decomposition is cyclic, as it is for dense tensors, which means clusters (blocks) of nonzeros are broken apart and mapped to different processors.
The cyclic distribution yields a load-balanced data layout for all practical scenarios we are aware of, but also sacrifices locality that could be potentially exploited by a more adaptive layout selection.

\subsection{Mapping Selection and Decomposition}

CTF employs nested SUMMA~\cite{Geijn:SUMMA:1997} and replication (2.5D and 3D algorithms) to execute tensor contractions on many processors~\cite{solomonik2014massively}.
For dense matrices, these algorithms are known to achieve asymptotically optimal communication costs for square matrices~\cite{SD_EUROPAR_2011} as well as for rectangular matrices~\cite{demmel2013communication}.
Dense nonsymmetric tensor contractions are isomorphic to matrix multiplication, which can be seen by mapping three sets of indices from the tensor contraction to the three indices involved in matrix multiplication (this is shown in detail in~\cite{SDH_ETHZ_2015}).
We note that this holds so long as each index appears in exactly two tensors, while CTF also supports indices appearing in one or all three tensors.
We focus on analyzing plain tensor contractions (each index in two tensors).
When an index appears in only one tensor, it can usually be separated into an independent summation operation (e.g., if the index is in one operand, so long as the multiplication operator is distributive), and this is desired for optimal cost.
Within summations, CTF handles such indices by either selecting a
mapping where the index is completely local (not decomposed over), or performing a broadcast/reduction for operands/\-outputs (both approaches have computation and communication cost no greater than the size of the larger tensor).
When an index appears in all three tensors, it defines a set of independent problems acting on disjoint datasets.
It is always desirable to distribute independent problems to different processors (although very small granularity of the problem or the overhead of redistribution can make other parallelizations more favorable).

CTF chooses the preferred tensor contraction algorithm, by tuning over a set of possible decompositions over a set of processor grids.
The processor grids are chosen by factorizing the processor count and forming all unique foldings of the factors into lower order processor grids, or by inferring the physical processor grid topology on machines such as BlueGene/P and BlueGene/Q.
So long as the processor count is factorizable, in particular, it has no large prime factors, this scheme yields a large space of potential decompositions.
For each processor grid, CTF considers possible assignments of the processor grid dimensions to tensor dimensions, as well as replication along the processor grid dimension, and evaluates the cost of the resulting algorithm.
When there are no large prime factors, it is easy to see that this space contains asymptotically optimal algorithms, which require processor grids with up to three dimensions.
Tuning over a larger space, permits lower constant factors on the communication cost, and the preservation of the initial decompositions of the tensors.
The details and analysis of these algorithms, decompositions, and necessary redistributions for dense tensors is studied in more detail in~\cite{solomonik2014massively}.

We employ the same algorithms and decomposition selection for sparse tensors, except only nonzeros are communicated.
The performance models that evaluate the cost of each decomposition, scale the computation, interprocessor communication, and memory bandwidth costs associated with each decomposition by the total fraction of nonzeros of the tensor (and by a tunable constant prefactor).
This scaling is valid so long as the number of nonzeros per processor in the decomposition is load balanced.
We first argue that this condition holds with high probability for a cyclic decomposition of a tensor with uniform sparsity.

Given a tensor with $z$ nonzeros, a contraction maps the tensor to a $n_1\times n_2$ matrix with $z$ nonzeros.
A cyclic decomposition of the tensor (or of the matrix), corresponds to an assignment of $n_1/q_1$ columns and $n_2/q_2$ rows to each of $q=q_1q_2$ processors ($q$ may be smaller than the total number of processors $p$ if the tensor contraction replicates the sparse tensor).
The intersection of these rows and columns, corresponds to $n_1n_2/q$ elements in the tensor, and since the nonzeros are uniformly distributed, by a classical balls-into-bins argument, so long as $z>q\log q$, there are $\Theta(z/q)$ nonzeros owned by each processor.

We also argue that if the rows and columns are randomized, the potential load imbalance cannot be too great.
The randomization can be done cheaply (with cost linear in the number of nonzeros at input/output time and no overhead at contraction time) by assigning the same random permutation to all tensor dimensions of the same length and applying this permutation to the indices of the elements the user reads and writes to the tensor.
The load imbalance is maximized, when the balls are heaviest, i.e., when there are dense rows or columns in the sparse matrix.
Without loss of generality, consider the case when $q_2> q_1$, $n_1>n_2$ and there are $c=\lfloor z/n_1\rfloor$ dense columns, the maximum number of zeros on any of the $q$ processors is with high probability,
\[x=\min\left(n_1n_2/q,z,(n_1/q_1)\cdot B(c,q_2)\right),\]
where $B(c,q_2)$ is the balls-into-bins bound with $c$ balls and $q_2$ bins, so when $c>q_2\log q_2$, $B(c,q_2)=\Theta(c/q_2)$ and $x=\Theta(z/q)$.
Thus, when the sparse tensor is distributed over a large enough number of processors in any cyclic layout, we can expect the layout to be load balanced.

In the current implementation, CTF does not do randomization of columns or rows.
The common application case would be for the sparse matrix to contain dense areas, for which a cyclic layout achieves ideal load balance.
Problems could only arise in peculiar scenarios, when for instance every $k$th column is dense.
We leave it for future work to study safety mechanisms to avoid exceeding memory constraints for such cases.
One approach would be to check the nonzero count of a given mapping after selection but prior to redistribution, and pick an alternative mapping or apply column/row permutations if the load balance is worse than an expected threshold.

\subsection{Algorithmic Cost Analysis}

Assuming, that the number of nonzeros is asymptotically perfectly load balanced, as shown under certain assumptions in the previous subsection, we can now analyze the communication cost of the contraction.
While CTF considers memory bandwidth, latency, and interprocessor communication costs, for brevity we only study the interprocessor communication cost in this paper.
Let $W$ be the greatest number of elements sent and received by any processor with memory $M$.
Every tensor contraction is reducible to a matrix multiplication $\B C=\B A\cdot \B B$ where $\B A$ is $m\times k$, $\B B$ is $k\times n$, and $\B C$ is $m\times n$.
We focus our analysis on the case when one of the tensors, $\B A$, is sparse, namely it contains $z$ entries for some $z\in[1,mk]$.
So long as our load balance assumption holds, the computation cost will be no greater than
\(F=\Theta(nz/p),\)
as in each computation step (step of SUMMA) the set of $\B A$ operands will correspond to those owned by some processor.

Since CTF selects initial tensor decompositions independently of the contractions that will be performed, there is an associated redistribution cost with most contractions.
Layouts that avoid this cost are specially considered and evaluated, but we consider the worst case, in which all tensors are redistributed (a load-balanced initial layout is always selected), which has a communication cost of 
\[W_\text{redist}=\Theta(z/p+kn/p+mn/p),\]
where we ignore the fact that the target layout may be a subset of processors, since in this case, the algorithm will replicate the tensor at the same communication cost as the redistribution.

As mentioned, 1D, 2D, and 3D decomposition are all in the space of CTF decompositions analyzed. The cost of the 1D decomposition corresponds to the full replication of one tensor, 
\[W_\text{1D}=\Theta(\min(z,kn,mn)).\]
In a 2D decomposition the SUMMA algorithm is used, and the three variants which keep one tensor local and communicate the other two~\cite{Geijn:SUMMA:1997} are considered for a cost of
\begin{align*}
W_\text{2D}=\Theta\bigg(\min_{p_1p_2=p}\big[\min(&z/p_1+kn/p_2, z/p_1+mn/p_2, 
kn/p_1+mn/p_2)\big]\bigg).
\end{align*}
In a 3D (2.5D) decomposition, one tensor is replicated over $p_3$ processors and a variant of SUMMA is used over the other two dimensions, which has the same total communication cost as replicating each tensor over one processor grid dimension:
\begin{align*}
W_\text{3D}=\Theta\bigg(\min_{p_1p_2p_3=p}\bigg[\frac{z}{p_1p_2}+\frac{kn}{p_1p_3}+\frac{mn}{p_2p_3}\bigg]\bigg),
\end{align*}
under the constraint that 
\[M\geq \min\bigg[\frac{hz}{p_1p_2},\frac{kn}{p_1p_3},\frac{mn}{p_2p_3}\bigg],\]
where $hz$ is the greatest number of elements of $\B A$ owned by any processor (by assumption/expectation $h$ is a constant).
Since the cheapest of all three possible layouts is chosen, the total cost for a contraction is
\[W\leq W_\text{redist} +\min(W_\text{1D},W_\text{2D},W_\text{3D}).\]
This reduces to simply the equation for the 3D decomposition cost,
\begin{align*}
W=\Theta\bigg(\min_{p_1p_2p_3=p}\bigg[\frac{z}{p_1p_2}+\frac{kn}{p_1p_3}+\frac{mn}{p_2p_3}\bigg]\bigg),
\end{align*}
under the constraint that,
\[M\geq \min\bigg[h\frac{z}{p_1p_2},\frac{kn}{p_1p_3},\frac{mn}{p_2p_3}.\bigg]\]
This cost is natural as one can imagine picking one or two of $p_1,p_2,p_3$ to be one (the redistribution cost makes the advantage of one tensor not being communicated at all in 1D and 2D algorithms asymptotically irrelevant).
When at least one of $p_1,p_2,p_3$ is $1$ (1D or 2D algorithm) the memory constraint is trivial as it is satisfied if all tensors fit into memory.

We now state and prove a lower bound on communication cost to try gauge how close to optimality the upper bound is.
\begin{theorem}
\label{thm:spmm_lb}
Consider any algorithm for matrix multiplication $\B C=\B A\cdot \B B$ where $\B A$ is $m\times k$ with $z$ nonzeros, $\B B$ is $k\times n$ (dense), and $\B C$ is $m\times n$.
Let $z_1=\min(m,\sqrt{z})$ and $z_2=\min(k,z/z_1)$, and define
 $r_1=\min(n,z_1,z_2)$, $r_2=\mathrm{median}(n,z_1,z_2)$, and $r_3=\max(n,z_1,z_2)$.
The algorithm must have a worst-case horizontal communication cost of at least $W=\bar{W}$, where
\[\bar{W}= 
\begin{cases}
\frac{r_1r_2r_3}{p\sqrt{M}} + 
\lt(\frac{r_1r_2r_3}p\rt)^{2/3} & : p>r_2r_3/r_1^2. \\
r_1\lt(\frac{r_2r_3}{p}\rt)^{1/2} & : r_2r_3/r_1^2 \geq p > r_3/r_2. \\ 
r_1r_2& : r_3/r_2 \geq p.
\end{cases}
\]
\end{theorem}
\begin{proof}
Consider an $A$ whose sparsity is such that the top right $z_1\times z_2$ block $\B{\bar{A}}$ is dense, where $z_1=\min(m,\sqrt{z})$ and $z_2=\min(k,z/z_1)$.
Let $\B{\bar{B}}$ be the first $z_2$ columns of $B$ and $\B{\bar{C}}$ be the first $z_1$ rows of $\B{C}$.
The multiplication $\B C=\B A\cdot \B B$ requires computation of the dense matrix multiplication $\B{\bar{C}}=\B{\bar{A}}\cdot \B{\bar{B}}$, and therefore the worst-case communication cost of the former cannot be lower than the latter.
By the dense matrix multiplication lower bound given in~\cite{demmel2013communication}, we obtain the communication cost lower bound stated in the theorem.
\end{proof}
We can observe that this lower bound is only attained by the provided algorithm when $n\geq p\max(z_1,z_2)$, the 1D layout where the sparse matrix is replicated over all processors.
In other cases, a matrix other than the sparse matrix is being communicated, and the lower bound reduction to dense matrix multiplication, permits the communicated dense matrix to be smaller by a factor of $\min(m/z_1,k/z_2)$.
Therefore, the question of whether a faster algorithm or a stronger lower bound for the problem exists, remains open.
However, based on the algorithmic communication cost upper-bound, we can observe that whenever the algorithm communicates the sparse tensor (which is advantageous whenever it is smaller or proportional in size to $\B B$ or $\B C$), the communication cost goes down proportionally with the number of nonzeros.
Since the computation cost of the contraction is also proportional to the number of nonzeros in $\B A$, we can expect that, at least for certain problems, the parallel scalability of the sparse algorithm will be as good as that of the dense algorithm.

\section{Performance Evaluation}
\label{sec:perf}

\begin{figure*}[t]
\centering
\subfigure[effect of sparsity on strong scaling performance]{
\includegraphics[width=3.1in]{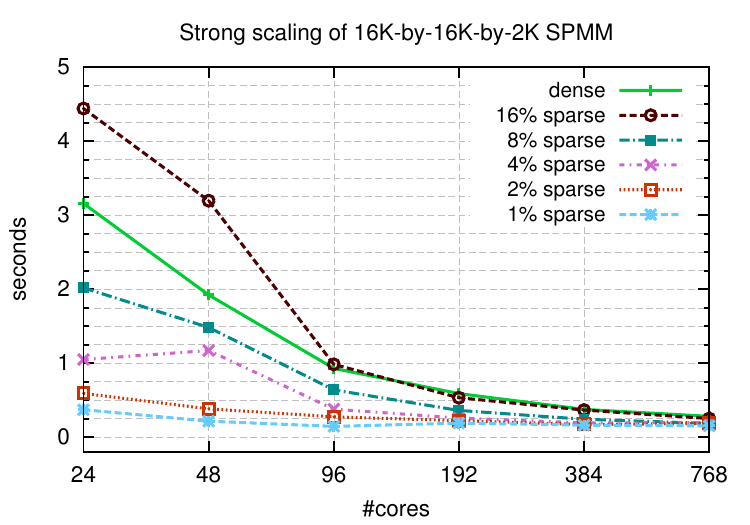}
\label{fig:spmm_ss_edison}
}
\subfigure[effect of sparsity on weak scaling performance]{
\includegraphics[width=3.1in]{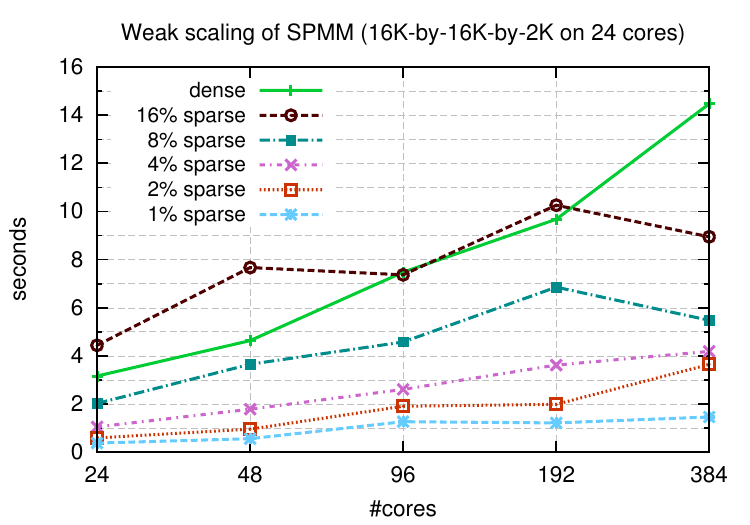}
\label{fig:spmm_ws_edison}
}
\caption{Performance of matrix multiplication of a square sparse matrix and a rectangular dense matrix on Edison (Cray XC30)}
\end{figure*}

We evaluate our algorithm and implementation for the contraction of sparse and dense tensors for three problems.
We start with a simple benchmark of multiplication of a sparse matrix by a somewhat smaller (in dimensions) dense matrix.
Then we move on to the MP3 method for electronic structure that contracts a sparse tensor with a dense tensor, a computationally similar problem to the first benchmark.
These two benchmarks operate on double precision floating point numbers and leverage the sparse column-sparse-row (CSR) matrix multiplication routines provided by the Intel MKL library (version 13.0.3).
Finally, we evaluate the performance of the framework when using only a reference sequential kernel, for APSP via path doubling.
These performance experiments test a limited range of the capabilities of the framework, but are representative of the main target applications.
In particular, we focus on sparse-matrix-matrix multiplication rather than sparse-matrix-vector multiplication, since in general tensor contractions reduce to the former, and additionally, the latter has already been studied extensively in existing literature.

Our experimental platform is Edison, a Cray XC30 architecture equipped
with two 12-core Intel ``Ivy Bridge" processors at 2.4GHz per node 
(19.2 Gflops per core and 460.8 Gflops per node)~\cite{edison}.
Edison has a Cray Aries high-speed interconnect with Dragonfly topology ($0.25 \mu s$ to $3.7 \mu s$ MPI latency, 8 GB/sec MPI bandwidth).
Each node has 64 GB of DDR3 1600 MHz memory (four 8 GB DIMMs per socket) and two shared 30 MB L3 caches (one per Ivy Bridge). 
Each core has its own L1 and L2 caches, of size 64 KB and 256 KB.

All of our benchmarks report the median time over ten iterations (excluding one warm-up iteration).
We did not observe significant variability in performance for all of the benchmarks.
Even when running at the strong scaling limit, performance outliers were within roughly 10\% of the median time.
This observation is expected, since the benchmarks have relatively large granularity and we did not use a significant fraction of the Edison system.


\subsection{Multiplication of a Dense Matrix by a Sparse Matrix}
Our first benchmark tests the scalability of the sparse functionality of CTF for one of the easiest use-cases.
We consider the matrix multiplication of a $n\times n$ sparse matrix containing a varied number of nonzeros, by a dense matrix with dimensions $n\times k$ (SPMM).
We pick $k$ to be a factor of $8$ smaller than $n$ in all experiments (varying $k$ over a larger range with respect to $n$ would be an interesting further experiment).
The nonzeros in the sparse matrix are selected based on an independent probability (thus the percentage of sparsity we report, is an expected and not an exact value).

Figure~\ref{fig:spmm_ss_edison} (the label (a) will always refer to the left figure and (b) to the right figure) demonstrates the strong scalability of the SPMM kernel with $n=16,384$ and $k=8,192$.
We first observe that while a benefit exists for sparse execution with respect to dense CTF kernels during single-node (24-core) execution, it is not proportional to the number of nonzeros. 
In fact, when the number of nonzeros is 16\%, $3.7$ seconds (85\% of the execution time) are spent on average in the local MKL CSR matrix multiplication (CSRMM) kernel, while the dense code spends only $2.5$ seconds.
However, when the sparsity percentage is decreased by a factor of two to 8\%, the CSRMM time falls by a factor of $2.57$.
Further decreases of the sparsity percentage by factors of two, yield to decreases of CSRMM time by less than a factor of two as expected, and generally diminish in their return.
This observation is consistent with the 3X improvement in performance of the 16\% sparse kernel when advancing from 48 to 96 cores.
The sequential overhead of switching to the sparse MKL routine is roughly a factor of ten, suggesting that a more efficient implementation may be possible.

The parallel strong scaling efficiency achieved by CTF for SPMM is better for higher counts of nonzeros, which is natural since the granularity of the parallelizable sequential work is higher.
When the number of cores reaches 768, sparsity is no longer as much of a benefit with respect to the dense code.
Even with 1\% sparsity, the speed-up over the dense code is only 1.8X rather than the 8.5X difference on 24 cores (in fact the lowest time to solution is achieved at 96 cores for the 1\% sparse case).

Figure~\ref{fig:spmm_ws_edison} shows the weak scaling of the SPMM kernel, which starts on 24 cores from the same problem as the strong scaling study, and increases both $n$ and $k$ by $\sqrt{2}$ whenever $p$ increases by two (keeping the memory usage per processor constant).
For weak scaling, sparsity is of substantial benefit on all core counts, with the speed-up obtained for the 1\% sparse problem increasing to 9.9X on 384 cores, and the time for the 16\% sparse problem becoming lower than the time needed for dense execution.
The contrast between the strong and weak scaling scenarios is highlighted by the time spent in the local contraction kernel (CSRMM), which for the strong scaled 1\% sparsity problem on 384 cores is only 13.5\%, while for the weak scaling 1\% sparse problem it is 44.5\% (on 384 cores the time in dense GEMM in 40.1\% for strong scaling 68.3\% for weak scaling for the dense execution).

\subsection{Sparse Electronic Structure Method (MP3)}

\begin{figure*}[t]
\centering
\subfigure[effect of sparsity on strong scaling performance]{
\includegraphics[width=3.1in]{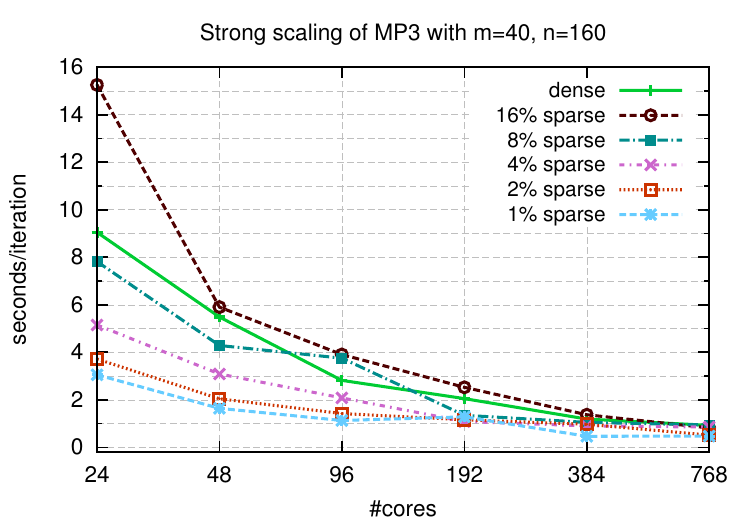}
\label{fig:spmp3_ss_edison}
}
\subfigure[effect of sparsity on weak scaling performance]{
\includegraphics[width=3.1in]{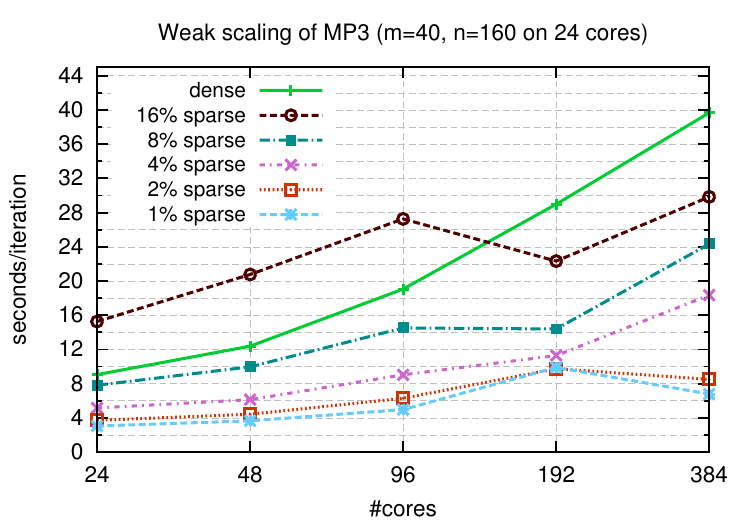}
\label{fig:spmp3_ws_edison}
}
\caption{Performance of MP3 for synthetic system of $m$ occupied orbitals and $n$ virtual orbitals with sparse integral tensors on Edison (Cray XC30)}
\end{figure*}

The SPMM kernel we analyzed is similar to the matrix multiplications that are isomorphic to the sixth order contractions presented in the MP3 calculation in Section~\ref{sec:apps:esc}.
However, computing the MP3 contractions via sparse matrix multiplication requires an extra step of transposing the tensors to the proper ordering.
Locally reordering the elements for the transposition step is relatively cheap with respect to the contraction, but in certain cases the transpositions also require data redistribution across processors, which is more expensive and less scalable.
Additionally, the MP3 method requires the application of a division, done as a custom user-defined function, and some other purely dense lower order contractions.

The strong and weak scaling results presented in Figures~\ref{fig:spmp3_ss_edison} and~\ref{fig:spmp3_ws_edison} nevertheless follow the same trends as the simpler SPMM kernel.
The baseline speed-up attained on 24 cores when the integral tensors are 1\% sparse is only 3.0X rather than the 8.5X obtained for SPMM, due to the additional operations.
However, the speed-up persists during strong scaling, with the 1\% sparse execution being 2.0X faster than the dense execution on 768 cores.
The parallel efficiency of the 1\% sparse code to 384 cores, where the lowest time to solution is achieved, is 41\%.
The weak scaling is done by increasing both $m$ and $n$ by $2^{1/4}$, every time $p$ is increased by two, keeping the size of all the integral tensors the same.
Since in this scaling regime the lower-order computations become less and less significant, the speed-up of the 1\% sparse version over the dense version increases to 5.8X on 384 nodes (from 3.0X on 24 cores). 
The percentage of time spent in MKL CSRMM on 384 cores is 10.5\% for strong scaling and 33.3\% for weak scaling (on 384 cores the time in dense GEMM in 34.8\% for strong scaling 63.3\% for weak scaling for the dense execution).

Overall, these results are positive for the prospect of accelerating electronic structure calculations via sparsity.
However, while our use of MKL CSRMM provides portability and leverages common infrastructure, the observed low sequential efficiency is a limitation of the current implementation.
We expect that more significant application speed-ups will be attainable from sparsity with the use of better-tuned sequential kernels.
A higher rate of performance may be achieved locally via instruction-level tuning or via the use of different sparse matrix formats.
When present, exploiting the presence of small dense blocks may be particularly beneficial, and has been used successfully for a different type of electronic structure method~\cite{Borštnik201447}.

\subsection{All-Pairs Shortest Paths}

\begin{figure*}[t]
\centering
\subfigure[effect of sparsity on strong scaling performance]{
\includegraphics[width=3.1in]{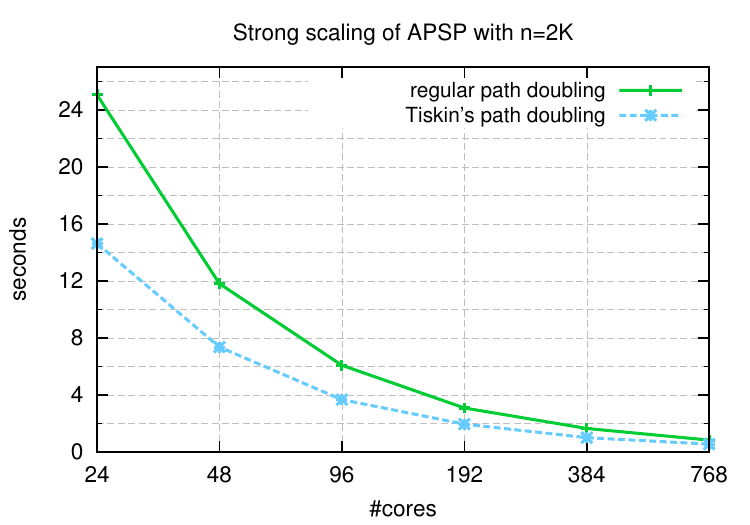}
\label{fig:apsp_ss_edison}
}
\subfigure[effect of sparsity on weak scaling performance]{
\includegraphics[width=3.1in]{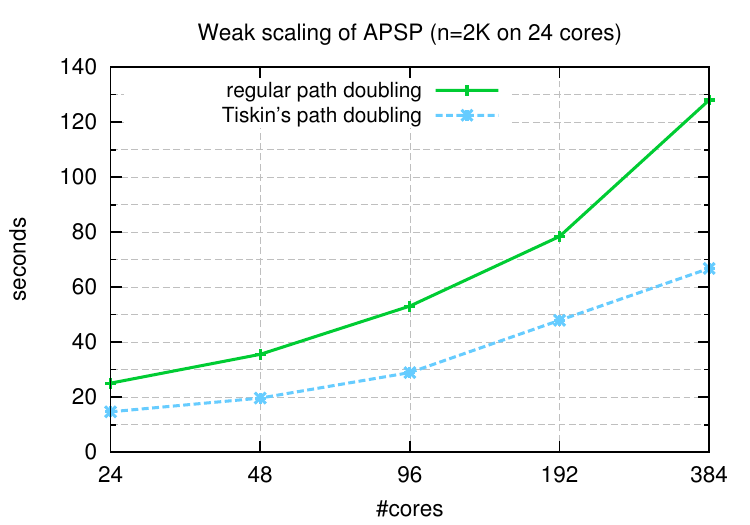}
\label{fig:apsp_ws_edison}
}
\caption{Performance of APSP on dense graph with $n$ nodes with random weights in range $[1,n^2]$ on Edison (Cray XC30)}
\end{figure*}

Lastly, we evaluate the performance of sparse and dense path doubling (introduced in Section~\ref{sec:apps:ga}).
Figure~\ref{fig:apsp_ss_edison} compares the strong scalability of the two algorithms for a dense graph of $n=2,048$ vertices, with edge weights selected as integers in the range $[1,n^2]$.
The overall performance of both kernels is quite low due to the use of a reference dense block integer multiplication kernel and a coordinate-format sparse integer multiplication kernel.
It is nevertheless interesting to observe that Tiskin's path doubling (for which we are the first to provide a parallel implementation) decreases the operation count and allows for noticeably better weak scalability than the dense kernel.
For weak scaling 87.8\% and 90.8\% percent of the time are spent in the local kernels on 384 cores for the dense and sparse algorithms, while for strong scaling we observe 86.6\% for the dense kernel and 71.4\% for the sparse kernel.

At this scale of parallelism, the lack of local optimized kernels is clearly still of overwhelming influence.
Our algebraic structure interface also enables the specification of user implementations of local sparse and dense matrix multiplication kernels.
Therefore, it is easy to plug-in a hand-tuned integer matrix multiplication kernel, however, for more complicated custom functions this may not always be available.
A more in-depth study of the problem that compares these methods with Floyd-Warshall for sparse graphs from application domains remains an interesting future work direction.

\section{Conclusion and Future Work}
\label{sec:conc}




In providing support for customizable sparse tensor operations in a high-performance distributed-memory context, we have extended the boundaries of the faculties of numerical libraries.
Under a minimalistic, yet powerful interface, we have hidden a flexible library that makes data-decomposition decisions at runtime, and employs state-of-the-art communication-avoiding algorithms.
The space of potential optimizations capable of accelerating this abstraction is vast, as evidenced by the extensive literature studying just the optimization of sparse-matrix-vector multiplication for floating point numbers, one of the simplest kernels generalized by our model.
Our performance results show the parallel scalability of our approach, despite only a limited amount of tuning having been done to date.

The multidimensionality of our work, which provides a programming model, an algorithmic study, a library implementation, and an evaluation on applications, enables numerous directions for future work.
The capabilities of the CTF framework can be further extended by providing support for contractions of two sparse tensor operands (SpGEMM~\cite{doi:10.1137/110848244}) and for output sparsity in contractions.
The performance and flexibility may also be improved by the study of optimizations and parallelism across multiple contractions, as well as the use of automatically-adjusted performance models.
The utility of our library is highlighted by the electronic structure application domain, where sparsity is an emerging solution to lowering the high-order complexity of methods for modeling electronic correlation.
In studying the performance of a simple method, MP3, we showed that sparsity can yield nearly a 6X speed-up and a benefit that increases with greater parallelism.
Further, we believe our approach will evolve to provide an elegant and efficient parallel programming model not only for electronic structure computations, but many applications in the greater domain of data-science.

\bibliographystyle{halpha}
\bibliography{paper}

\newcommand{\etalchar}[1]{$^{#1}$}
\begin{thebibliography}{PMvdG{\etalchar{+}}13}

\bibitem[ABB{\etalchar{+}}92]{LAPACK}
E.~Anderson, Z.~Bai, C.~Bischof, J.~Demmel, J.~Dongarra, J.~D. Croz,
  A.~Greenbaum, S.~Hammarling, A.~McKenney, S.~Ostrouchov, and D.~Sorensen.
\newblock {\em {LAPACK {U}sers' {G}uide}}.
\newblock SIAM, Philadelphia, PA, USA, 1992.

\bibitem[BAB{\etalchar{+}}05]{1386652}
G.~Baumgartner, A.~Auer, D.~Bernholdt, A.~Bibireata, V.~Choppella, D.~Cociorva,
  X.~Gao, R.~Harrison, S.~Hirata, S.~Krishnamoorthy, S.~Krishnan, C.~Lam,
  Q.~Lu, M.~Nooijen, R.~Pitzer, J.~Ramanujam, P.~Sadayappan, and A.~Sibiryakov.
\newblock Synthesis of high-performance parallel programs for a class of ab
  initio quantum chemistry models.
\newblock {\em Proceedings of the IEEE}, 93(2):276 --292, February 2005.

\bibitem[Bat94]{batagelj1994semirings}
V.~Batagelj.
\newblock Semirings for social networks analysis.
\newblock {\em Journal of Mathematical Sociology}, 19(1):53--68, 1994.

\bibitem[BBD{\etalchar{+}}13]{Ballard:2013:COP:2486159.2486196}
G.~Ballard, A.~Buluc, J.~Demmel, L.~Grigori, B.~Lipshitz, O.~Schwartz, and
  S.~Toledo.
\newblock Communication optimal parallel multiplication of sparse random
  matrices.
\newblock In {\em Proceedings of the Twenty-fifth Annual ACM Symposium on
  Parallelism in Algorithms and Architectures}, SPAA '13, pages 222--231, New
  York, NY, USA, 2013. ACM.

\bibitem[BCC{\etalchar{+}}97]{SCALAPACK}
L.~S. Blackford, J.~Choi, A.~Cleary, E.~{D'Azevedo}, J.~Demmel, I.~Dhillon,
  J.~Dongarra, S.~Hammarling, G.~Henry, A.~Petitet, K.~Stanley, D.~Walker, and
  R.~C. Whaley.
\newblock {\em {ScaLAPACK Users' Guide}}.
\newblock SIAM, Philadelphia, PA, USA, May 1997.

\bibitem[Bel58]{bellman1958routing}
R.~Bellman.
\newblock {On a Routing Problem}.
\newblock {\em Quarterly of Applied Mathematics}, 16:87--90, 1958.

\bibitem[BFF{\etalchar{+}}09]{Buluc:2009:PSM:1583991.1584053}
A.~Bulu\c{c}, J.~T. Fineman, M.~Frigo, J.~R. Gilbert, and C.~E. Leiserson.
\newblock Parallel sparse matrix-vector and matrix-transpose-vector
  multiplication using compressed sparse blocks.
\newblock In {\em Proceedings of the Twenty-first Annual Symposium on
  Parallelism in Algorithms and Architectures}, SPAA '09, pages 233--244, New
  York, NY, USA, 2009. ACM.

\bibitem[BG11]{bulucc2011combinatorial}
A.~Bulu{\c{c}} and J.~R. Gilbert.
\newblock The {C}ombinatorial {BLAS}: Design, implementation, and applications.
\newblock {\em International Journal of High Performance Computing
  Applications}, 25:496--509, 2011.

\bibitem[BG12]{doi:10.1137/110848244}
A.~Bulu{\c{c}} and J.~R. Gilbert.
\newblock Parallel sparse matrix-matrix multiplication and indexing:
  Implementation and experiments.
\newblock {\em SIAM Journal on Scientific Computing}, 34(4):C170--C191, 2012.

\bibitem[BKSE12]{2012arXiv1209.5145B}
J.~{Bezanson}, S.~{Karpinski}, V.~B. {Shah}, and A.~{Edelman}.
\newblock {Julia: A Fast Dynamic Language for Technical Computing}.
\newblock {\em ArXiv e-prints}, September 2012, 1209.5145.

\bibitem[Bra01]{brandes2001faster}
U.~Brandes.
\newblock A faster algorithm for betweenness centrality.
\newblock {\em Journal of Mathematical Sociology}, 25(2):163--177, 2001.

\bibitem[BS75]{BarlettMP3_1975}
R.~J. Bartlett and D.~M. Silver.
\newblock Many‐body perturbation theory applied to electron pair correlation
  energies. i. closed‐shell first‐row diatomic hydrides.
\newblock {\em The Journal of Chemical Physics}, 62(8):3258--3268, 1975.

\bibitem[BVWH14]{Borštnik201447}
U.~Borštnik, J.~VandeVondele, V.~Weber, and J.~Hutter.
\newblock Sparse matrix multiplication: The distributed block-compressed sparse
  row library.
\newblock {\em Parallel Computing}, 40(5–6):47 -- 58, 2014.

\bibitem[CLSH15]{ChowHF2015}
E.~Chow, X.~Liu, M.~Smelyanskiy, and J.~R. Hammond.
\newblock Parallel scalability of {H}artree--{F}ock calculations.
\newblock {\em The Journal of Chemical Physics}, 142(10):--, 2015.

\bibitem[CLV15]{2015arXiv150900309C}
J.~A. {Calvin}, C.~A. {Lewis}, and E.~F. {Valeev}.
\newblock {Scalable Task-Based Algorithm for Multiplication of
  Block-Rank-Sparse Matrices}.
\newblock {\em ArXiv e-prints}, September 2015, 1509.00309.

\bibitem[DEF{\etalchar{+}}13]{demmel2013communication}
J.~Demmel, D.~Eliahu, A.~Fox, S.~Kamil, B.~Lipshitz, O.~Schwartz, and
  O.~Spillinger.
\newblock Communication-optimal parallel recursive rectangular matrix
  multiplication.
\newblock In {\em IEEE International Symposium on Parallel Distributed
  Processing (IPDPS)}, 2013.

\bibitem[DG08]{dean2008mapreduce}
J.~Dean and S.~Ghemawat.
\newblock Mapreduce: simplified data processing on large clusters.
\newblock {\em Communications of the ACM}, 51(1):107--113, 2008.

\bibitem[edi]{edison}
{NERSC} description of {Edison} configuration.
\newblock
  \url{https://www.nersc.gov/users/computational-systems/edison/configuration}.
\newblock Accessed: 2015-09-12.

\bibitem[EWK{\etalchar{+}}13]{JCC:JCC23377}
E.~Epifanovsky, M.~Wormit, T.~Kuś, A.~Landau, D.~Zuev, K.~Khistyaev,
  P.~Manohar, I.~Kaliman, A.~Dreuw, and A.~I. Krylov.
\newblock New implementation of high-level correlated methods using a general
  block tensor library for high-performance electronic structure calculations.
\newblock {\em Journal of Computational Chemistry}, 34(26):2293--2309, 2013.

\bibitem[FF58]{ford1958network}
L.~Ford and D.~Fulkerson.
\newblock Network flow and systems of representatives.
\newblock {\em Canad. J. Math}, 10(1):78--84, 1958.

\bibitem[Flo62]{Floyd:1962}
R.~W. Floyd.
\newblock Algorithm 97: Shortest path.
\newblock {\em Commun. ACM}, 5:345--, June 1962.

\bibitem[GL05]{BGL2005}
D.~Gregor and A.~Lumsdaine.
\newblock The {P}arallel {BGL}: A generic library for distributed graph
  computations.
\newblock {\em Parallel Object-Oriented Scientific Computing (POOSC)}, 2:1--18,
  2005.

\bibitem[Hir03]{doi:10.1021/jp034596z}
S.~Hirata.
\newblock Tensor {C}ontraction {E}ngine: Abstraction and automated parallel
  implementation of configuration-interaction, coupled-cluster, and many-body
  perturbation theories.
\newblock {\em The Journal of Physical Chemistry A}, 107(46):9887--9897, 2003.

\bibitem[KG11]{doi:10.1137/1.9780898719918}
J.~Kepner and J.~Gilbert.
\newblock {\em Graph Algorithms in the Language of Linear Algebra}.
\newblock Society for Industrial and Applied Mathematics, 2011.

\bibitem[KM13]{Kats_sp_tensor2013}
D.~Kats and F.~R. Manby.
\newblock Sparse tensor framework for implementation of general local
  correlation methods.
\newblock {\em The Journal of Chemical Physics}, 138(14):--, 2013.

\bibitem[KRS89]{Kruskal1989135}
C.~P. Kruskal, L.~Rudolph, and M.~Snir.
\newblock Techniques for parallel manipulation of sparse matrices.
\newblock {\em Theoretical Computer Science}, 64(2):135 -- 157, 1989.

\bibitem[LHKK79]{lawson1979basic}
C.~L. Lawson, R.~J. Hanson, D.~R. Kincaid, and F.~T. Krogh.
\newblock Basic {L}inear {A}lgebra {S}ubprograms for {F}ortran usage.
\newblock {\em ACM Transactions on Mathematical Software (TOMS)},
  5(3):308--323, 1979.

\bibitem[LSR{\etalchar{+}}13]{Lai:2013:FLB:2503210.2503290}
P.-W. Lai, K.~Stock, S.~Rajbhandari, S.~Krishnamoorthy, and P.~Sadayappan.
\newblock A framework for load balancing of tensor contraction expressions via
  dynamic task partitioning.
\newblock In {\em Proceedings of SC13: International Conference for High
  Performance Computing, Networking, Storage and Analysis}, SC '13, pages
  13:1--13:10, New York, NY, USA, 2013. ACM.

\bibitem[MAB{\etalchar{+}}10]{malewicz2010pregel}
G.~Malewicz, M.~H. Austern, A.~J. Bik, J.~C. Dehnert, I.~Horn, N.~Leiser, and
  G.~Czajkowski.
\newblock Pregel: a system for large-scale graph processing.
\newblock In {\em Proceedings of the 2010 ACM SIGMOD International Conference
  on Management of Data}, pages 135--146. ACM, 2010.

\bibitem[MP34]{moller1934note}
C.~M{\o}ller and M.~S. Plesset.
\newblock Note on an approximation treatment for many-electron systems.
\newblock {\em Physical Review}, 46(7):618, 1934.

\bibitem[NHL96]{springerlink_ga1}
J.~Nieplocha, R.~J. Harrison, and R.~J. Littlefield.
\newblock Global {A}rrays: A nonuniform memory access programming model for
  high-performance computers.
\newblock {\em The Journal of Supercomputing}, 10:169--189, 1996.

\bibitem[PHG10]{doi:10.1080/00268971003662896}
J.~A. Parkhill and M.~Head-Gordon.
\newblock A sparse framework for the derivation and implementation of fermion
  algebra.
\newblock {\em Molecular Physics}, 108(3-4):513--522, 2010.

\bibitem[PMvdG{\etalchar{+}}13]{elemental1}
J.~Poulson, B.~Marker, R.~A. van~de Geijn, J.~R. Hammond, and N.~A. Romero.
\newblock Elemental: A new framework for distributed memory dense matrix
  computations.
\newblock {\em ACM Trans. Math. Softw.}, 39(2):13:1--13:24, February 2013.

\bibitem[Pri57]{6773228}
R.~Prim.
\newblock Shortest connection networks and some generalizations.
\newblock {\em Bell System Technical Journal, The}, 36(6):1389--1401, Nov 1957.

\bibitem[RNL{\etalchar{+}}13]{rajbhandari2013framework}
S.~Rajbhandari, A.~Nikam, P.-W. Lai, K.~Stock, S.~Krishnamoorthy, and
  P.~Sadayappan.
\newblock Framework for distributed contractions of tensors with symmetry.
\newblock {\em Preprint, Ohio State University}, 2013.

\bibitem[SD11]{SD_EUROPAR_2011}
E.~Solomonik and J.~Demmel.
\newblock Communication-optimal parallel {2.5D} matrix multiplication and {LU}
  factorization algorithms.
\newblock In {\em Euro-Par 2011 Parallel Processing}, volume 6853 of {\em
  Lecture Notes in Computer Science}, pages 90--109. Springer Berlin
  Heidelberg, 2011.

\bibitem[SDH15]{SDH_ETHZ_2015}
E.~Solomonik, J.~Demmel, and T.~Hoefler.
\newblock Communication lower bounds for tensor contraction algorithms.
\newblock Technical report, ETH Z\"urich, 2015.

\bibitem[SMH{\etalchar{+}}14]{solomonik2014massively}
E.~Solomonik, D.~Matthews, J.~R. Hammond, J.~F. Stanton, and J.~Demmel.
\newblock A massively parallel tensor contraction framework for coupled-cluster
  computations.
\newblock {\em Journal of Parallel and Distributed Computing},
  74(12):3176--3190, 2014.

\bibitem[SS95]{HF_screening_1995}
D.~L. Strout and G.~E. Scuseria.
\newblock A quantitative study of the scaling properties of the
  {H}artree--{F}ock method.
\newblock {\em The Journal of Chemical Physics}, 102(21):8448--8452, 1995.

\bibitem[SW01]{LCCSD_2001}
M.~Sch{\"u}tz and H.-J. Werner.
\newblock Low-order scaling local electron correlation methods. {IV}. linear
  scaling local coupled-cluster ({LCCSD}).
\newblock {\em The Journal of Chemical Physics}, 114(2):661--681, 2001.

\bibitem[Tis01]{tiskin_apsp}
A.~Tiskin.
\newblock All-pairs shortest paths computation in the {BSP} model.
\newblock {\em Lecture Notes in Computer Science, Automata, Languages and
  Programming}, 2076:178--189, 2001.

\bibitem[Val90]{valiant1990bridging}
L.~G. Valiant.
\newblock A bridging model for parallel computation.
\newblock {\em Communications of the ACM}, 33(8):103--111, 1990.

\bibitem[vdGW97]{Geijn:SUMMA:1997}
R.~A. van~de Geijn and J.~Watts.
\newblock {SUMMA}: {S}calable {U}niversal {M}atrix {M}ultiplication
  {A}lgorithm.
\newblock {\em Concurrency: Practice and Experience}, 9(4):255--274, 1997.

\bibitem[War62]{Warshall:1962}
S.~Warshall.
\newblock A theorem on boolean matrices.
\newblock {\em J. ACM}, 9:11--12, January 1962.

\end{thebibliography}

\end{document}